\documentclass[12pt,journal,draftcls,letterpaper,onecolumn]{IEEEtran}
\usepackage{natbib}
\usepackage{amsthm}
\usepackage{amssymb}
\usepackage{amsmath}
\input{amssym.def}
\input{amssym}
\usepackage[pagebackref=true,colorlinks,linkcolor=blue,citecolor=magenta,bookmarksnumbered=true]{hyperref}
\usepackage{epsfig}
\usepackage{epstopdf}
\usepackage{subfigure}
\usepackage{subref}
\usepackage{color}

\newtheorem{theorem}{Theorem}
\newtheorem{lemma}{Lemma}

\newtheorem{proposition}{Proposition}
\newtheorem{definition}{Definition}

\begin{document}
%
\title{TripNet: A Heuristic Algorithm for Constructing Rooted Phylogenetic Networks from Triplets}

\author{ H. Poormohammadi,%
      \and
         Ch. Eslahchi,%
     
       \and
        and R. Tusserkani%
     
      }

\maketitle

\IEEEcompsocitemizethanks{\IEEEcompsocthanksitem H. Poormohammadi and Ch. Eslahchi are with department of Mathematics, Shahid Beheshti University, G.C. Tehran, Iran .\protect\\
E-mail: see http://www.ch-eslahchi@sbu.ac.ir
\IEEEcompsocthanksitem R. Tusserkani is with school of Mathematics, Institute for Research in Fundamental Sciences (IPM), Tehran, Iran.}

\begin{abstract}
The problem of constructing an optimal rooted phylogenetic 
network from a set of rooted triplets is an NP-hard problem. 
In this paper, we present a heuristic algorithm called TripNet
which tries to construct an optimal rooted phylogenetic network 
from an arbitrary set of triplets.
We prove some theorems to justify the performance of the algorithm. 
\end{abstract}


\section{Introduction}

Phylogenetic networks are a generalization of phylogenetic trees that permit
the representation of non-tree-like underlying histories. A rooted
phylogenetic network is a rooted directed acyclic graph 
in which no nodes has indegree greater than 2 and the outdegree of each node with indegree 2 is 1.
Such nodes are called reticulation nodes.
Mathematicians are interested in developing methods
that infer a phylogenetic tree or network from basic building blocks. In the
computation of a rooted tree or network, one group of the basic building blocks are triplets, the rooted binary trees on
three taxa \cite{hussonbook}. 
In 1981, Aho et al., studied the problem of constructing a
tree from a set of triplets \cite{AhoSagiv1}. They proposed an algorithm called the BUILD algorithm which shows that,
given a set of triplets, it is possible to construct in polynomial
time a rooted tree that all the input triplets are contained in
it or decide that no such tree exists. 
When there is no tree for a given set of triplets one may try to produce an optimal phylogenetic network.
In this context, the goal is to compute an optimal rooted phylogenetic
network that contains all the rooted triplets. One possible optimality criterion is to minimize the {\it level} of the network, which
is defined as the maximum number of reticulation nodes contained in any biconnected
component of the network. The other optimality criterion is to minimize the number of reticulation nodes \cite{hussonbook}.
In \cite{galled7} and \cite{Levelone} the authors considered the
problem of deciding whether, given a set of triplets as input, is
it possible to construct a level-1 phylogenetic network that contains all the input triplets? They showed
that, in general, this problem is NP-hard. However,
in \cite{Levelone} the authors showed that when the set of triplets is dense, which means that
for each set of three taxa there is at least one triplet in the input set, 
the problem can be solved in polynomial time.
After their results, all research in this new
area has up to this point focused on constructing networks from
dense triplet sets.
The algorithm by \cite{lev2optimum} can be used to
find a level-1 or a level-2 phylogenetic network which minimizes the number of
reticulation nodes if such a network exists. In \cite{levelk} the authors showed that given a dense set of triplets
$\tau$ and a fixed number $k$, it is possible to construct in time $O(|\tau|^{k+1})$ a level-$k$ phylogenetic network
consistent with $\tau$ or decides that no such network exists.

In this paper we present a heuristic algorithm called {\it TripNet} for
constructing phylogenetic networks from an arbitrary set of triplets.
Despite of current methods that work for dense set of triplets, 
a key innovation is the applicability of TripNet to non-dense set of triplets.
The results of the TripNet algorithm on biological sequences is presented in \cite{tripnetpaper}.
Here we prove some theorems to justify the performance of the algorithm. 
 This paper is organized as follows. In section II we present some
 definitions and notation. In section III 
  we discuss triplet construction methods.
 In section IV the  
 directed graph $G_{\tau}$ related to a set of triplets $\tau$ is introduced.
 In section V the concept of the height function of a tree is introduced,
and we propose an algorithm to construct a tree from its
height function.
Then we generalize the concept of the height function to the networks.
Finally, in section VI we present the TripNet algorithm.

\section{Definitions and Notation}

Let $X$ be a set of taxa. A {\it rooted phylogenetic tree} ({\it tree} for short) on $X$ is a rooted 
unordered leaf labeled tree whose leaves are distinctly labeled
by $X$ and every node which is not a leaf has at least outdegree two. 
A {\it directed acyclic graph} (DAG) is a directed graph that is free of directed cycles.
A directed acyclic graph $G$ is {\it connected} if there is an undirected path between any two
nodes of $G$. It is {\it biconnected} if it contains no node whose removal disconnects 
$G$. A biconnected component of a graph $G$ is a maximal biconnected subgraph of $G$.
A {\it rooted phylogenetic network} ({\it network} for short) on $X$ is a rooted directed acyclic graph in which
{\it root} has indegree 0 and outdegree 2 and
every node except the root satisfies one of the following
conditions:
\\ a) It has indegree 2 and outdegree 1. These
nodes are called {\it reticulation nodes}.
\\ b) It has indegree 1
and outdegree 2.
\\ c) It has indegree 1 and outdegree 0. These nodes are called {\it leaves} and are distinctly labeled by $X$.

A {\it reticulation leaf} is a leaf whose parent is a reticulation node.
A network is said to be a {\it level-$k$ network } if each of its biconnected
component contains at most $k$ reticulation nodes.
A tree can be considered as a level-0 network.

A {\it rooted triplet} ({\it triplet} for short) is a binary rooted unordered tree with three leaves. We
use $ij|k$ to denote the triplet with taxa $i$ and $j$ on one side
and $k$ on the other side of the root (Fig.\ref{fig1}). 
A set of triplets $\tau$ is called {\it dense} if for each subset of three taxa,
there is at least one triplet in $\tau$.
A triplet $ij|k$ is {\it consistent} with a 
network $N$ or equivalently $N$ is consistent with $ij|k$ if $N$
contains a subdivision of $ij|k$, i.e. if $N$ contains distinct
nodes $u$ and $v$ and pairwise internally node-disjoint paths
$u \rightarrow i$, $u \rightarrow j$, $v \rightarrow u$ and 
$v\rightarrow k$. Fig. \ref{fig2}  shows an example of a
network which is consistent with $ij|k$. A set
$\tau$ of triplets is consistent with a network $N$
if all the triplets in $\tau$ are consistent with $N$.
We use the symbols $\tau(N)$ and $L_N$ to represent the set of all triplets that are consistent
with $N$ and the set of labels of its leaves respectively. For any set $\tau$ of triplets define $L(\tau)=  \cup_{t \in \tau} {L_t}$.
The set $\tau$ is called a set of triplets on $X$ if $L(\tau)=X$.

\begin{figure}[h]
\centering \subfigure[]{
\includegraphics[scale=0.4]{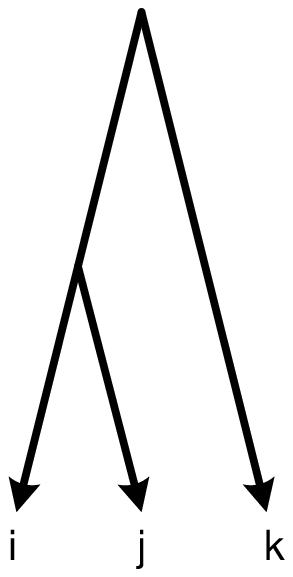}
\label{fig1} } \subfigure[$N$]{
\includegraphics[scale=.5]{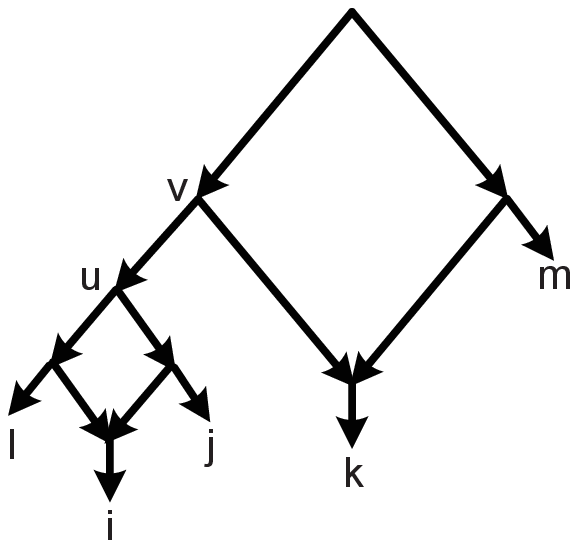}
\label{fig2} } \subfigure[]{
\includegraphics[scale=0.5]{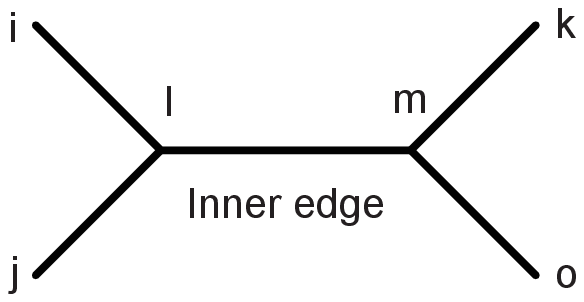}
\label{quartet} } \subfigure[]{
\includegraphics[scale=0.5]{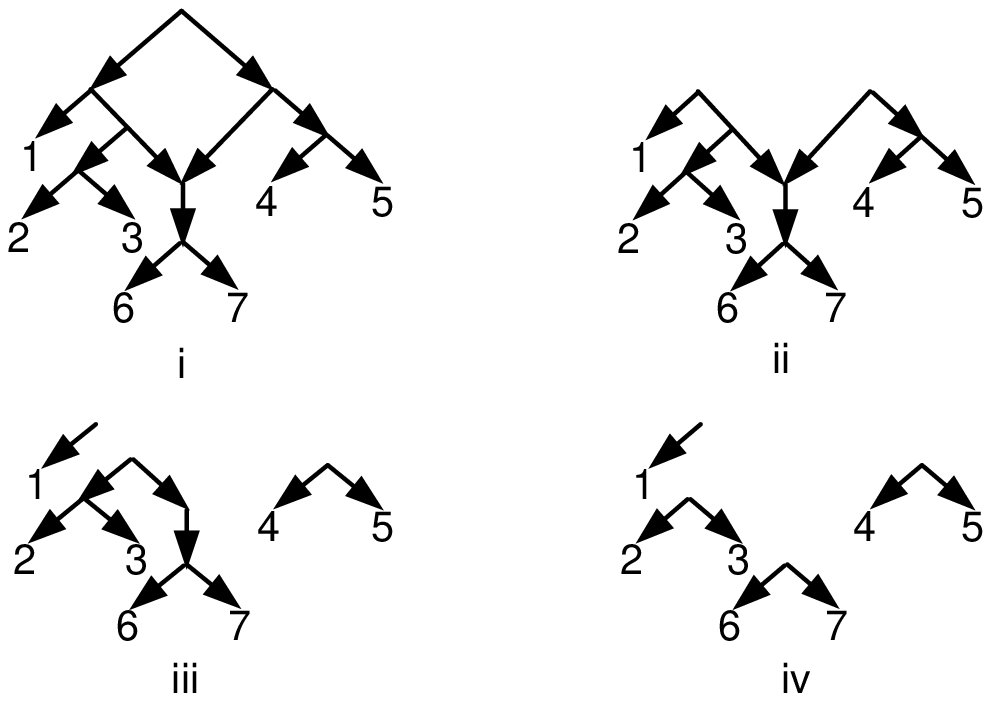}
\label{snfindexamp} } 
\subfigure[]{
\includegraphics[scale=0.5]{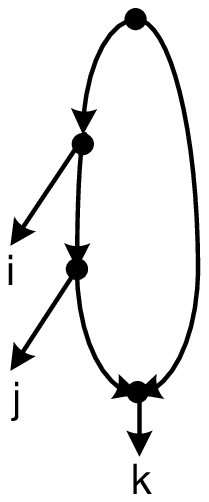}
\label{counterexam} } 
 \label{fig1v2} \caption{ (a) Triplet $ij|k$.
(b) Triplet $ij|k$ is consistent with the network
$N$. (c) The steps of removing edges with maximum weight
from a network. (d) quartet $ij|ko$ with its inner edge. 
(e) A counter example for the reverse of Theorem \ref{thorm3}. }
\end{figure}

\section{Triplets construction method}

There are two main tree construction methods,
character-based methods and distance-based methods.
In character-based methods, the information of a set $X$ of biological
sequences is directly used for producing final tree.
In distance-based methods first a distance matrix $D$ is
computed from $X$ and then a rooted (or unrooted) tree $T$ is constructed
from $D$ \cite{hussonbook}.

A weighted tree $(T, w)$ is a rooted (or unrooted) tree $T$ together with
a function $w : E(T ) \rightarrow \mathbb{R}$. We call $w(e)$ the 
weight of the edge $e$.
For any two nodes $i$ and  $j$ of $T$,
let $l_{ij}$ denotes the unique path in $T$ from $i$ to $j$. Define
\begin{eqnarray*}
d_{T}(i, j)= \sum\limits_{e \, \in \, l_{ij} }  {w(e)}.
\end{eqnarray*}

If $T$ is an unweighted tree then we suppose that for each edge $e$ in $T$, $w(e)=1$.
Given a set of taxa $X$, let $(T, w)$ be a weighted tree on $X$ and
$D_{T}$ be a matrix in which the entry of row $i$ and column $j$ is $d_{T}(i, j)$.
We call $D_{T}$ the distance matrix related to $(T, w)$.

A {\it quartet} is a binary
unrooted tree with four leaves. We use the symbol
$ij|kl$ for a quartet on the set of taxa $\{i, j, k, l\}$ which
have neighbor pairs $i, j$ and $k, l$. In a quartet $Q$ there is a
unique edge such that its two end points are not leaves. 
We call this edge the {\it inner edge} of $Q$ $($ Fig. \ref{quartet} $)$. 
A weighted quartet is called $informative$, if the weight of its inner edge is positive. 
The following proposition holds for informative quartets.

\,\,\,\,\,\,\,
\begin{proposition}
\label{propos1} Given a set of four taxa $X=\{i, j, k, l \}$ and a distance matrix $D$ on $X$. 
 For an informative quartet $ij|kl$, the equation
$d(i, j)+d(k, l) < d(i, k) + d(j, l) = d(j, k) + d(i, l)$ holds.
\end{proposition}

\,\,\,\,\,\,\,

The following proposition also holds \cite{fourqunique}.

\,\,\,\,\,\,\

\begin{proposition}
\label{propos2} Given a set of four taxa $X$ and a distance matrix $D$ on $X$. 
There is a unique quartet $Q$ on $X$ for which $D_{Q}=D$.
\end{proposition}

\,\,\,\,\,\,\,

Suppose that $X$ is a set of taxa in which each taxon is a biological sequence. Let $o_{X}$ 
be an outgroup for $X$ and $D$ be a distance matrix on $X \cup \{o_{X}\}$.
In this paper to obtain a set of triplets, we use the method
introduced in \cite{lev2optimum}. For each set of four taxa that contains
 $o_{X}$, we construct the unique quartet which satisfies
Proposition \ref{propos2}. Then, by removing $o_{X}$ from informative quartets we obtain the set of triplets.
In the rest of paper we refer to this method By QOT.

\section{The directed graph related to a set of triplets}

Let $\tau$ be a set
of triplets. Define $G_{\tau}$, the directed graph related to $\tau$,
 by $V(G_{\tau})= \{ ~ \{i, j\} : i,j \in L(\tau), i \neq j \} $ (we denote $\{i, j\}$ by $ij$
for short) and $E(G_{\tau})= \{ (ij,ik)  :  ij|k \in \tau \}  \cup \{ (ij,jk)  :  ij|k \in \tau \}$. 
The graph $G_{\tau}$ has an important role in the remaining of the paper and in this section we prove some basic
properties of $G_{\tau}$.

Let $X$ be a set of sequences, $D=[d(i, j)]$ be a distance matrix on $X$
where for any pair $i, j \in X$,  $d(i, j)$ denotes the distance between them, and $\tau$ be the set of
triplets that is produced by QOT
method. 
Here we define the concept of the closure of $\tau$.
If $i j|k$ and $j s|i$ are in $\tau$, then we have the quartets $i j|k o_{X}$
and $j s|i o_{X}$. According to the Proposition \ref{propos1},
$d(i, j) + d(k, o_{X}) < d(j, k) + d(i, o_{X})$,
$d(j, s) + d(i, o_{X}) < d(i, j) + d(s, o_{X})$, and therefore
$d(j, s) + d(k, o_{X}) < d(j, k) + d(s, o_{X})$.
It means that we should have the quartet $j s|k o_{X}$.
The equivalent triplet for this informative quartet  is $j s|k$.
If this triplet is not in $\tau$ add it to $\tau$ and
 continue this procedure until one cannot add more
triplets. We use the symbol $\overline{\tau}$ to show this new set of
triplets and call it the closure of $\tau$.

The following lemma is an immediate consequence of the definition of $\overline{\tau}$.

\,\,\,\,\,\,\,

\begin{lemma}
\label{taubarunique} Let $X$ be a set of sequences and $\tau$ be the set of
triplets which is produced by the QOT
method. Then $\overline{\tau}$ contains at most one triplet for each
$\{i, j ,k\} \subseteq X$.
\end{lemma}

\,\,\,\,\,\,\,

Now we state the main results of this section.

\,\,\,\,\,\,\,

\begin{theorem}
\label{thorm4} Let $X$ be a set of sequences and $\tau$ be the set of
triplets which is produced by the QOT method. 
Then $G_{\tau}$ is a DAG.
\end{theorem}

\begin{proof}
We prove a stronger result and show
that $G_{\overline{\tau}}$ is a DAG.
The proof proceeds by induction on the length of the shortest cycle in $G_{\overline{\tau}}$.
First we prove that 
$G_{\overline{\tau}}$ dose not contain any cycle
of length 3. Assume that $C$ is a cycle of length 3 in $G_{\overline{\tau}}$.
 Let $(ij, ik)$ be an edge of $C$. The triplet which corresponds to
this edge is $ij|k$. Suppose that the third node of the cycle is
$st$. Thus the other edges of the cycle are 
$(st, ij)$ and $(ik, st)$. So $|\{i, j\} \cap \{s, t\}|=|\{i, k\} \cap \{s, t\}|=1$.  
There are two cases. 
Case 1: $s=i$ and
$t \neq j, k$ (or $t=i$ and $s \neq j, k$). 
Case 2: $s=j$ and $t=k$ (or $s=k$ and $t=j$). For
the first case, the edges of $C$ are $(ij, ik)$,
$(ik, it)$ and $(it, ij)$. The three quartets which are
corresponds to the triplets of these three edges are $ij|ko_{X}$, $ik|to_{X}$ and $it|jo_{X}$. According to
Proposition \ref{propos1}, we have three inequalities
$d(i, j) + d(k, o_{X}) < d(i, k) + d(j, o_{X})$,
$d(i, k) + d(t, o_{X}) < d(i, t) + d(k, o_{X})$ and
$d(i, t) + d(j, o_{X}) < d(i, j) + d(t, o_{X})$. By summing up these 
inequalities, we obtain a contradiction. For the
second case 
the edges of $C$ are $(ij, ik)$, $(ik, jk)$ and
$(jk, ij)$. The three triplets corresponds to these three edges are  
$i j|k, i k|j$ and $j k|i$ which
contradicts Lemma \ref{taubarunique}.
So there is no cycle of length 3 in $G_{\overline{\tau}}$. 
Now assume that there is no cycle of length $k \geq 3$ in $G_{\overline{\tau}}$
and $C$ be a cycle of length $k+1$ in it.
First we claim that there
is no path $s_{1}s_{2}\rightarrow s_{3}s_{4} \rightarrow s_{5}s_{6} $ in $C$
such that $|\{s_{1},s_{2}\} \cap \{s_{3},s_{4}\} \cap \{s_{5},s_{6}\}|=1$.
Suppose that there exists such a path. So this path is of
the form $js \rightarrow ji \rightarrow jk$ and triplets $js|i$
and $ij|k$ are in $\overline{\tau}$. The method of constructing
$\overline{\tau}$ implies that $js|k$ is in $\overline{\tau}$ and
the edge $js \rightarrow jk$ is in $G_{\overline{\tau}}$. 
So we obtain a cycle of length $k$ in
$G_{\overline{\tau}}$, a contradiction. 

Let $s_{1}s_{2}$ be a node of $C$.
There exists a node $s_{3}$ such that the edge
$s_{1}s_{2}\rightarrow s_{1}s_{3}$ is in $C$ and
$s_{1}s_{3}$
is connected to a node $s_{1}s_{4}$ or $s_{3}s_{4}$.
If $s_{1}s_{4} \in V(C)$ then
$s_{1}s_{2}\rightarrow s_{1}s_{3} \rightarrow s_{1}s_{4}$  
will be in $C$ which contradicts the above claim.
So the cycle $C$ is of the form 
$s_{1}s_{2}
\rightarrow s_{1}s_{3} \rightarrow s_{3}s_{4} \rightarrow
 \ldots \rightarrow
 s_{k}s_{k+1} \rightarrow s_{k+1}s_{k+2}
\rightarrow s_{1}s_{2}$. 
For the edges $s_{k+1}s_{k+2} \rightarrow s_{1}s_{2}$ we obtain $|\{s_{1},s_{2}\} \cap \{s_{k+1},s_{k+2}\}|=1$. For the cases 
$s_{k+1}=s_{l}, l \in \{1, 2\} $ or $s_{k+2}=s_{1}$ we have a
cycle of length $k$ in $G_{\overline{\tau}}$. So
$s_{k+2}=s_{2}$ and triplets $s_{1}s_{2}|s_{3}$,
$s_{1}s_{3}|s_{4}$, $s_{3}s_{4}|s_{5}$, 
$\ldots$, $s_{k-1}s_{k}|s_{k+1}$, $s_{k}s_{k+1}|s_{2}$ and
$s_{k+1}s_{2}|s_{1}$ are in
$\overline{\tau}$. Equivalently, we have the following inequalities: \\
$d(s_{1},s_{2})+d(s_{3},o)<d(s_{1},s_{3})+d(s_{2},o)$, \\
$d(s_{1},s_{3})+d(s_{4},o)<d(s_{3},s_{4})+d(s_{1},o)$, \\
\vdots \\
$d(s_{k},s_{k+1})+d(s_{1},o)<d(s_{k+1},s_{1})+d(s_{k},o)$, \\
$d(s_{k+1},s_{1})+d(s_{2},o)<d(s_{1},s_{2})+d(s_{k+1},o)$. \\
Summing these inequalities, we obtain a
contradiction. So there is no
cycle of length $k+1$ in $G_{\overline{\tau}}$.
\end{proof}

\,\,\,\,\,\,

Let $\tau$ be a set of triplets that is consistent with a tree. 
Let $T_{\tau}$ denotes the unique tree that is produced by the BUILD algorithm.

\,\,\,\,\,\,\,

\begin{theorem}
\label{thormgtaudag} Let $\tau$ be a set of triplets that is consistent with a tree.
Then $G_{\tau}$ is a DAG.
\end{theorem}

\begin{proof} The proof proceeds by induction on $|L(\tau)|$.
It is trivial when $|L(\tau)|=3$.  Assume that theorem holds when $|L(\tau)| \leq k$.
Let $|L(\tau)|=k+1$ and $T_{1}, T_{2}, \ldots, T_{m}$ be
$m$ subtrees that are obtained from $T_{\tau}$ by removing its root.
For each $i$, $1 \leq i \leq m$, let $\tau_{i}=\tau |_{L_{T_i}}$ denotes
the set of all triplets in $\tau$ whose leaves are in $L_{T_{i}}$. 
By the induction assumption for each $i$, $1 \leq i \leq m$,
 $G_{\tau_{i}}$ is a DAG. Let $\tau'= \bigcup\limits_{1 \leq i \leq m} {\tau_{i}}$ and
$G'= \bigcup\limits_{1 \leq i \leq m} {G_{\tau_{i}}}$. Apparently, $G'$ is a DAG and $G_{\tau'}=G'$.
 The graph $G_{\tau}$ can be obtained from $G'$ by adding the nodes which belong to  $V(G_{\tau}) \setminus V(G')$
 and the edges corresponds to the triplets in $\tau \setminus \tau'$.
 If a triplet $t=ab|c \in \tau \setminus \tau'$ then there is $1 \leq i < j \leq m$ such that $a,b \in L(\tau_i)$ and $c \in L(\tau_j)$.
 It means that the edges corresponds to the triplets in $\tau \setminus \tau'$ are of the form $(ab, ac)$ such that $ac \in V(G_{\tau}) \setminus V(G')$.
 So all nodes in $V(G_{\tau}) \setminus V(G')$ has outdegree zero and the edges in $G_{\tau} \setminus G'$ are from $V(G')$ to  $V(G_{\tau}) \setminus V(G')$.
 Now if there exists a cycle in $G_{\tau}$ it sholud contain a node
 in $V(G_{\tau}) \setminus V(G')$ which contradicts that these nodes have outdegree zero
and the proof is complete.
\end{proof}

\,\,\,\,\,\,\,

\section{Height function}

In this section first the concept of the height
function of a tree and a DAG is introduced and  then the BUILD algorithm is restated based on
this concept.

Let $X \choose 2$ denotes the set of all subsets of $X$ of size $2$.

\begin{definition}
\label{heghttree} Let $X$ be an arbitrary finite set.
A function $h: {X \choose 2} \rightarrow \mathbb{N}$ is called a {\it height function} on $X$.
\end{definition}

Let $T$ be a rooted tree with the root $r$, $c_{ij}$ be the lowest common ancestor of the leaves
$i$ and $j$, and $l_T$ denotes the length of the longest path started from $r$. 

\begin{definition}
\label{heghttree}
The {\it height function} of $T$, $h_{T}$ is defined as $h_{T}(i,j)=l_{T}-d_{T}(c_{ij},r)$ where $i$ and $j$ are two distinc leaves of $T$.
\end{definition}

The following theorem represents the relation between the height function of a tree and a triplet consistent with it.

\,\,\,\,\,\,\,
\begin{theorem}
\label{thormhandtree} Let $T$ be a tree. A triplet $ij|k$ is consistent with $T$
if and only if  $h_{T}(i, j) < h_{T}(i, k)$ or $h_{T}(i, j) < h_{T}(j, k)$.
\end{theorem}

\begin{proof} Let $ij|k$ be consistent with $T$. 
By definition $h_{T}(i, j) < h_{T}(i, k)$ and $h_{T}(i, j) < h_{T}(j, k)$.
Now suppose that for the three arbitrary leaves $i, j$ and $k$, we have $h_{T}(i, j) < h_{T}(i, k)$
or $h_{T}(i, j) < h_{T}(j, k)$. Without loss of generality suppose that $h_{T}(i, j) < h_{T}(i, k)$.
Since $c_{ij}$ and $c_{ik}$ are on a unique path from the root $r$ to $i$ and $d_{T}(c_{ij}, r) > d_{T}(c_{ik}, r)$, thus
there is a path from the lowest common ancestor of $i, k$ 
to the lowest common ancestor of $i, j$ which follows that $ij|k$ is consistent with $T$.
\end{proof}

\,\,\,\,\,\,\,

Let $\tau$ be a set of triplets, $G_{\tau}$ be a DAG and $l_{G_{\tau}}$ denotes the length of the longest path in $G_{\tau}$.
Since $G_{\tau}$ is a DAG, the set of nodes with outdegree zero is nonempty.
Assign $l_{G_{\tau}}+1$ to the nodes with
outdegree zero and remove them from $G_{\tau}$. Assign $l_{G_{\tau}}$ to the nodes with
outdegree zero in the resulting graph and continue this procedure until all nodes are removed. 

\begin{definition}
\label{dagdefine}
For any two distinct
$i, j \in L(\tau)$, define $h_{G_{\tau}}(i,j)$ as the value that is assigned by the above procedure to the
node $ij$ and call it the {\it height function} related to $G_{\tau}$. 
\end{definition}

Let $\tau$ be a set of triplets that is consistent with a tree. By Theorem \ref{thormgtaudag}, $G_\tau$
is a DAG and $h_{G_{\tau}}$ is well-defined.
The following theorem represents a method to obtain $h_{{T}_{\tau}}$ from $\tau$ using $h_{G_{\tau}}$.

\,\,\,\,\,\,\,
\begin{theorem}
\label{thorm2hequal} Let $\tau$ be a set of triplets which is consistent with a tree.
Then $h_{G_{\tau}}=h_{T_{\tau}}$.
\end{theorem}

\begin{proof} The proof proceeds by induction on $|L_{T_{\tau}}|$.
It is trivial when $|L_{T_{\tau}}|=3$. Assume that theorem holds when
 $|L_{T_{\tau}}| \leq k $.
Let $|L_{T_{\tau}}|=k+1$ and $T_{1}, T_{2}, \ldots, T_{m}$ be
$m$ subtrees which are obtained from $T_{\tau}$ by removing its root.
For each $i$, $1 \leq i \leq m$ let $\tau_{i}=\tau |_{L_{T_i}}$ 
, and $r_{i}$ be the root of $T_{i}$. 
By the induction assumption for each $i$, $1 \leq i \leq m$,
$h_{G_{\tau_{i}}}=h_{T_{\tau_{i}}}$.
Moreover we conclude from the BUILD algorithm that $T_{i}=T_{\tau_{i}}$, for $1 \leq i \leq m$. 
Thus $h_{G_{\tau_{i}}}=h_{T_{i}}$, for $1 \leq i \leq m$. 
So for $i$, $1 \leq i \leq m$, the maximum length of the longest path in $T_{i}$ is $l_{T_{\tau}}-1$.
It means that for $i$, $1 \leq i \leq m$, the maximum length of the longest path in $G_{\tau_{i}}$ is $l_{T_{\tau}}-2$.
Therefore by the proof of Theorem \ref{thormgtaudag} the length of the longest path in $G_{\tau}$ is $l_{T_{\tau}}-1$.
Let $a, b \in L_{T_{\tau}}$. We have two cases. 

{\bf Case 1}:
For some $i$ and $j$, $1 \leq i < j \leq m$, $a \in L_{T_{i}}$ and $b \in L_{T_{j}}$. 
Since the outdegree of $ab$ in
$G_{\tau}$ is zero and $c_{ab}=r$, then
$h_{T_{\tau}}(a, b)=l_{T_{\tau}}=h_{G_{\tau}}(a,b)$.

{\bf Case 2}:
For some $i$, $1 \leq i \leq m$, $a,b \in L_{T_{i}}$. By the induction assumption $h_{G_{\tau_{i}}}(a,b)=h_{T_{\tau_{i}}}(a,b)$. 
Therefore
$h_{T_{\tau}}(a,b)=l_{T_{\tau}}-d_{T_{\tau}}(c_{ab},r)=l_{T_{\tau}}-(d_{T_{\tau_{i}}}(c_{ab},r_{i})+1)=(l_{T_{\tau}}-l_{T_{\tau_{i}}}-1)+(l_{T_{\tau_{i}}}-d_{T_{\tau_{i}}}(c_{ab},r_{i}))=$
$ (l_{T_{\tau}}-l_{T_{\tau_{i}}}-1)+h_{T_{\tau_{i}}}(a,b)=(l_{T_{\tau}}-l_{T_{\tau_{i}}}-1)+h_{G_{\tau_{i}}}(a,b)=h_{G_{\tau}}(a,b)$.
The last equality is obtained by construction of $G_{\tau}$ from $G_{\tau_{i}}$ which is stated in the proof of Theorem \ref{thormgtaudag}.
So for each $a, b \in L_{T_{\tau}}$, $h_{T_{\tau}}(a, b)=h_{G_{\tau}}(a,b)$ and the proof is complete. 
\end{proof}

\,\,\,\,\,\,\,

Now we describe an algorithm similar to BUILD algorithm, using height function. We refer to this algorithm by {\it HBUILD}.
Let $h$ be a height function on $X$.
Define a weighted complete graph $(G,h)$ where $V(G)=X$
and edge $\{i, j\}$ has weight $h(i,j)$. 
Remove the edges with maximum weight from $G$.
If removing these edges results a connected graph the algorithm stops. Otherwise, the process
of removing the edges with maximum weight is continued in each connected component until each connected component
contains only one node.
At the end of this procedure one can 
reconstruct the tree by reversing the steps of the algorithm similar to BUILD algorithm
(see Fig. 2). 
The algorithm above decides in polynomial time whether a
tree with height function $h$ exists. 

Now if $\tau$ is a set of triplets which is consistent with a tree,
by 
Theorems \ref{thormgtaudag} and \ref{thorm2hequal}, $G_{\tau}$ is a DAG and
$h_{G_{\tau}}=h_{T_{\tau}}=h$ and HBUILD algorithm constructs $T_{\tau}$.

\begin{figure}[h]
\centering \subfigure[]{
\includegraphics[scale=0.5]{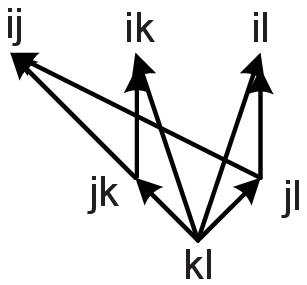}
\label{gtau} } \subfigure[]{
\includegraphics[scale=.5]{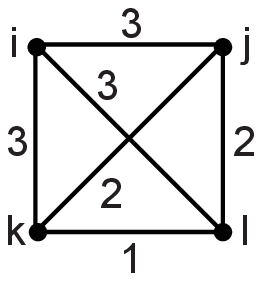}
\label{gh} } \subfigure[]{
\includegraphics[scale=.5]{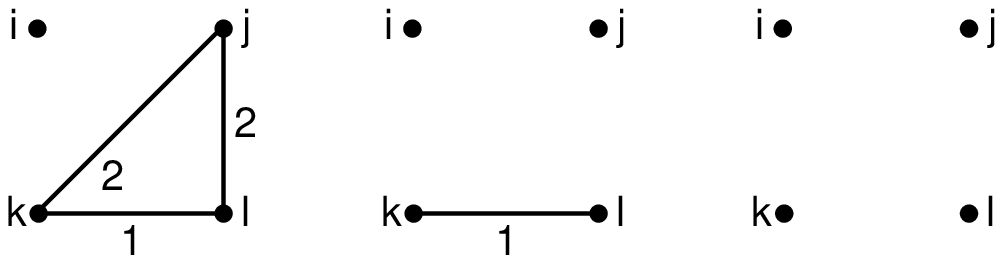}
\label{ghremoveall} }  \subfigure[]{
\includegraphics[scale=0.5]{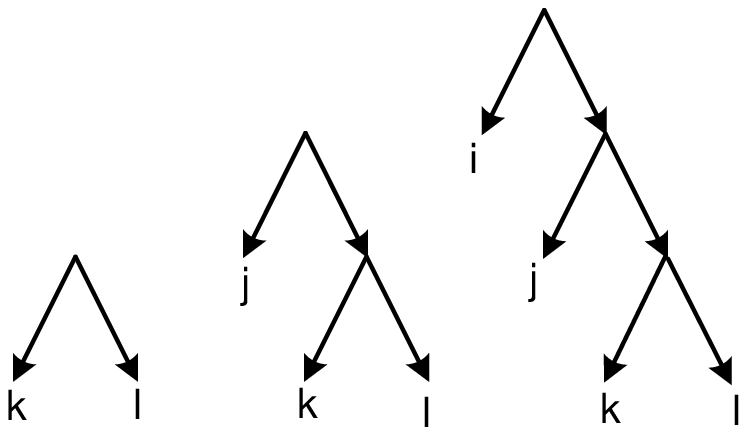}
\label{ghtreeall} } 
 \label{gtughtre} \caption {The steps of constructing $T_{\tau}$ from the given set $\tau=\{kl|j, kl|i, jk|i, jl|i\}$. 
 (a) The graph $G_{\tau}$. (b) The graph $(G,h)$. (c) Removing maximum weights from the graph $(G,h)$. 
 (d) Constructing $T_{\tau}$ using step c. } 
\end{figure}

Now we generalize the concept of height function from trees to networks.
This generalization is not straightforward because the concept of (lowest) common ancestor of
two leaves of a network is not well-defined.
Let $N$ be a network with the root $r$ and $l_N$ be
the length of the longest directed path from $r$ to the leaves. 
For each node $u$ consider $d(u,r)$ as the length of the longest directed path from $r$ to $u$. 
For any two nodes $u$ and $v$, we call $u$ an {\it ancestor} of $v$,
if there exists a directed path from $u$ to $v$. If $u$ is an ancestor of $v$ then we say that $v$ is {\it lower} than $u$.
{\it Lowest common ancestor} of two leaves in a network is not necessarily unique. 
For any two leaves $i$ and $j$, let $C_{ij}$
be the set of all lowest common ancestors of $i$ and $j$.

\,\,\,\,\,\,\,

\begin{definition}
\label{heightnetwork}
For each pair of leaves $i$ and $j$, define $h_{N}(i,j)= min \{l_{N}-d(c,r) : c \in C_{ij} \}$
and call it the {\it height function} of $N$.
\end{definition}

Obviously, every network $N$ indicates a unique
height function $h_N$. But two different networks may have the same height function
(see Fig. \ref{sameheight}) .

\begin{figure}[h]
\centering   \subfigure[$T$ ~~~~~~ $N$]{
\includegraphics[scale=.5]{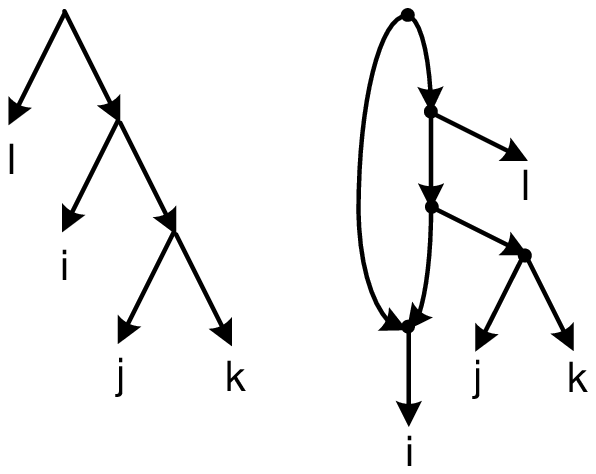}
\label{sameheight} } \subfigure[$T_1$ ~~~~~~~~~~~~~~ $T_2$]{
\includegraphics[scale=.5]{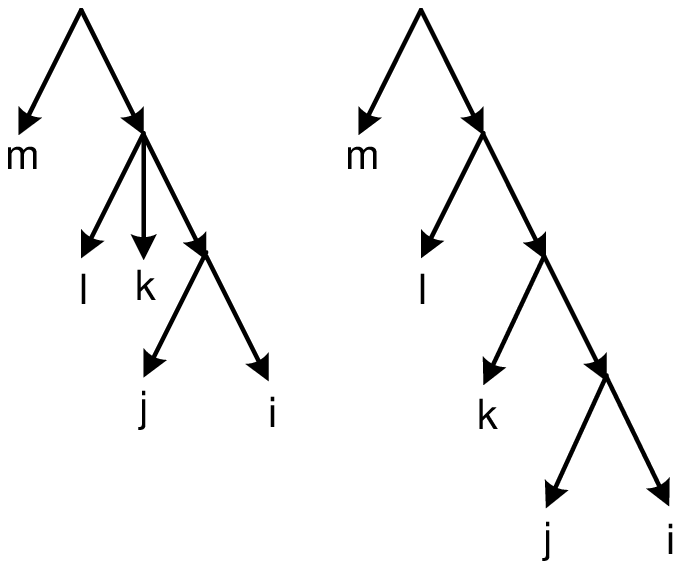}
\label{binarization} }
 \label{sameheightnetdif} \caption{  (a) Two different networks with the same height function. $h_{N}=h_{T}=h$.   
 $h(j,k)=1, h(i,j)=h(i,k)=2 ~$  and $~ h(i,l)=h(j,l)=h(k,l)=3$. (b) $T_2$ is a binarization of $T_1$.}
\end{figure}

In the following proposition we prove that for a given height function $h$ there
is a network $N$ such that $h_N=h$.

\,\,\,\,\,\,\,

\begin{proposition}
\label{heightindicateanet} Let $X$ be an arbitrary finite set and
$h$ be a height function on $X$. Then there 
exists a network $N$ not necessarily binary, such that $h_N=h$.
\end{proposition}

\begin{proof} Let $h_{max}= max \{h\{x, y\}  :  x, y \in X\}$.
Let $r$ be the root of $N$. For each pair of nodes $x$ and $y$ with
$h(x, y) = h_{max}$, we connect $x$ and $y$ to $r$.
For each pair of nodes $x$ and $y$ with $h(x, y) < h_{max}$, we consider a new node and
connect this node to $r$ by $h_{max}-h(x, y)$ edges.
By deleting multiple edges we obtain a network $N$ with $h_{N}=h$.
\end{proof}

\,\,\,\,\,\,\,

The following theorem shows the relation between height function of a network
and the triplet consistency with it.

\,\,\,\,\,\,\,

\begin{theorem}
\label{thorm3} Let $N$ be a network, $i, j$, and $k$ be its three distinct leaves.
If $h_{N}(i, j) < h_{N}(i, k)$ or $h_{N}(i, j) < h_{N}(j, k)$ then $ij|k$ is consistent with $N$.
\end{theorem}

\begin{proof} Suppose that $h_{N}(i, j) < h_{N}(i, k)$. 
Let $v_{ij}$ and $v_{ik}$ be common ancestors of $i, j$ and $i, k$ respectively,
such that $h_{N}(i, j) = l_{N}-d(v_{ij}, r)$ and $h_{N}(i, k) = l_{N}-d(v_{ik}, r)$. Let $l_{i}$ and $l_{j}$ be two distinct paths from
$v_{ij}$ to $i$ and $j$, respectively. Let $l_{k}$ be an arbitrary path from $v_{ik}$ to $k$.
If $l_{i} \cap l_{k} \neq \emptyset$ then it follows that $h_{N}(i,j) \geq h_{N}(i,k)$ which is a contradiction. So $ij|k$ is consistent with $N$.
\end{proof}

\,\,\,\,\,\,\,

The reverse of the above theorem is not necessarily true.
For example, consider the network of Fig. \ref{counterexam} .
The triplet $ij|k$ is consistent with it, but $h(i,j)=h(i,k)=3$ and $h(j,k)=2$.

The basic idea of the TripNet algorithm is to find a height
function as an intermediate computational step that yields the
minimum amount of information required to construct the network from a set of triplets.
So it is important to find a way for computing $h_{N}$ from a set of
triplets. In the rest of this section we introduce a computational method for computing $h_{N}$
using Integer Programming.
Let $\tau$ be a set of triplets with $|L(\tau)|=n$.
Inspiring from Theorems \ref{thormhandtree} and \ref{thorm3}, for each triplet $ij|k \in \tau$, define two inequalities $h(i,k)-h(i,j) \geq 1$ and
$h(j,k)-h(i,j) \geq 1$. Since the number of variables in such inequalities are at most $c(n,2)$, we obtain the following
system of inequalities from $\tau$.
\begin{eqnarray*}
&&h(i, k) - h(i, j) \ge 1 \,\,\,\,\,\,\,\,\,  ij|k \in \tau, \\
&&h(j, k) - h(i, j) \ge 1 \,\,\,\,\,\,\,\,\,  ij|k \in \tau, \\
&&0 < h(i, j) \le c(n, 2) \,\,\,\,\,\,\,\,\, 1 \leq i, j \leq n.
\end{eqnarray*}
%
Let $s$ be an integer. Define the following Integer Programming and call it IP($\tau, s$).
\begin{eqnarray*}
\text{Maximize}&&\sum_{1 \leq i, j \leq n} {h(i, j)}, \\
\text{Subject to}:&&,h(i, k) - h(i, j) \ge 1~~~~~~ij|k \in \tau, \\
&&h(j, k) - h(i, j) \ge 1~~~~~~ij|k \in \tau, \\
&&0 < h(i, j) \le s ~~~~~~~~~~~~1 \leq i, j \leq n.
\end{eqnarray*}
%
%
%
%
%
%

Intuitively if IP($\tau, s$) has a feasible solution, we expect that the optimal
solution to this integer programming is an approximation of the height function of
an optimal network $N$ consistent with $\tau$. The following theorems support this intuition.

\,\,\,\,\,\,\,

\begin{theorem}
\label{thormipsmaxforl} Let $\tau$ be a set of triplets. 
Then $G_\tau$ is a DAG if and only if for some integer $s$, the IP($\tau, s$) has a feasible solution.
In this case the minimum number $s$, for which
IP($\tau, s$) has a feasible solution, is $\l_{G_{\tau}}+1$.  
\end{theorem}

\begin{proof} 
Let $G_{\tau}$ be a DAG. 
Without loss of genrerality assume that $G_{\tau}$ is connected.
The proof proceeds by induction on $l_{G_{\tau}}$. 
If $l_{G_{\tau}}=1$ then obviously for $s=1$, IP($\tau, s$) has no feasible solution
and for each $s \geq 2$, IP($\tau, s$) has a feasible solution.
Assume that the theorem holds for $l_{G_{\tau}} \leq k$. 
Suppose that $\tau$ is a set of triplets with $l_{G_{\tau}}=k+1$. 
Let $A$ be the set of the terminal nodes of all longest paths in $G_{\tau}$.
For each $ij \in A$ there is some $x \in L(\tau)$ such that $ix|j \in \tau$. 
Let $B$ be the set of all such triplets and $\tau'=\tau \setminus B$.
Apparently, $B \neq \emptyset$ and the length of the longest path in $G_{\tau'}$ is $k$.
By the induction assumption the minimum number $s$ for which
IP($\tau', s$) has a feasible solution, is $\l_{G_{\tau'}}+1=\l_{G_{\tau}}$.  
Consider IP($\tau, \l_{G_{\tau}}+1$). Define $h(i, j)=\l_{G_{\tau}}+1$, for each $ij \in A$ and $h(t,l)=h'(t,l)$, for each $tl \notin A$. 
$h$ is a feasible solution to IP($\tau, \l_{G_{\tau}}+1$).
Now if $s$ be a solution for IP($\tau, s$) then $s-1$ is a solution for IP($\tau', s-1$).
So $\l_{G_{\tau}}+1$ is minimum soltion for IP($\tau, s$).
Now suppose that $\tau$ is a set of triplets  
and for some integer $s$, IP($\tau, s$) has a feasible solution $h$.
Assume that $G_{\tau}$ has a cycle
$C=i_{1} j_{1} \rightarrow i_{2} j_{2} \rightarrow \ldots \rightarrow i_{m} j_{m} \rightarrow i_{1} j_{1}$.
Corresponds to $C$ we have inequalities $h(i_{1},j_{1}) < h(i_{2},j_{2}) < \ldots <h(i_{m},j_{m}) <h(i_{1},j_{1})$
which is a contradiction
and the proof is complete.
\end{proof}

 \,\,\,\,\,\,\,

Let $\tau$ be a set of triplets consistent with a tree.
By Theorems \ref{thormgtaudag}, \ref{thormhandtree}, and \ref{thormipsmaxforl},  
$h_{T_{\tau}}$ is a feasible solution to IP($\tau, l_{G_{\tau}}+1$).
In the following theorem we prove the uniqueness of this solution.

\,\,\,\,\,\,\,

\begin{theorem}
\label{thormipfortuniquemax} Let $\tau$ be a set of triplets consistent with a tree.
Then $h_{T_{\tau}}$ is the unique optimal solution to IP($\tau, l_{G_{\tau}}+1$).
\end{theorem}

\begin{proof} 
By Theorem \ref{thormgtaudag}, $G_{\tau}$ is a DAG. So $l_{G_{\tau}}$ is well defined.
The proof proceeds by induction on $l_{G_{\tau}}$. Without loss of genrerality assume that $G_{\tau}$ is connected.
The theorem is trivial when $l_{G_{\tau}}=1$. Let for each set of triplets consistent with a tree,
$h_{T_{\tau}}$ be the unique optimal solution to 
IP($\tau, l_{G_{\tau}}+1$) where $l_{G_{\tau}}=k \geq 1$.
Suppose that $\tau$ is a set of triplets consistent with a tree and $l_{G_{\tau}}=k+1$.
Let $\tau'$ be the set of triplets which is introduced in the proof of Theorem \ref{thormipsmaxforl}. 
By the induction assumption
$h_{T_{\tau'}}$ is the unique optimal solution to IP($\tau', l_{G_{\tau'}}+1$).
By Theorem \ref{thormipsmaxforl} the minimum $s$ for which IP($\tau, s$)
has a feasible solution is  $l_{G_{\tau}}+1$.
Also $l_{G_{\tau'}}+1=l_{G_{\tau}}$. It follows that  
$h_{T_{\tau}}$ is the unique optimal solution to the IP($\tau, l_{G_{\tau}}+1$) and
the proof is complete.
\end{proof}

\,\,\,\,\,\,\,

The BUILD tree is not necessarily a binary tree. To obtain
a binary tree consistent with a set of triplets we do the following procedure.

Let $T$ be a tree and $x$ be a node of $T$ with $x_1, x_2, \ldots, x_k$, $k \geq 3$ as its childs.
Consider a new node $y$. Construct $T'$ by removing the edges $(x, x_1), (x, x_2), \ldots, (x, x_{k-1})$ from $T$
and adding the edges $(x, y), (y, x_1), (y, x_2), \ldots, (y, x_{k-1})$ to $T$.
Continuing the same method for each node with outdegree more than 2 we obtain a binary tree 
which we call it a {\it binarization} of $T$ (see Fig. \ref{binarization}).
Obviously, we can obtain different binarization of $T$.
The proof of the following theorem is easy and we omit it.
 
 \,\,\,\,\,\,\,
\begin{theorem}
\label{propos3} Let $\tau$ be a set of triplets that is consistent with a tree $T_1$, and $T_2$
be a binarization of $T_1$. Then $\tau$ is consistent with $T_2$. 
\end{theorem}

 \,\,\,\,\,\,\,

\section{TripNet algorithm}

Now we describe the TripNet algorithm in nine steps. 
In this algorithm the input is a set of triplets $\tau$ and the output is a network consistent with $\tau$.
Also if $\tau$ is consistent with a tree the algorithm constructs a binarization of $T_{\tau}$.
\\
\\
{\bf Step 1}: In this step we find a height function $h$
on $L(\tau)$. 
If $G_{\tau}$ is a DAG we set $G'_{\tau}=G_{\tau}$.
If $G_{\tau}$ is not a DAG we remove some edges from
$G_{\tau}$ in such a way that the resulting graph $G'_{\tau}$ is a
DAG. Set $h=h_{G'_{\tau}}$. 
\\
\\
If $\tau$ is obtained
from biological sequences using the QOT method, then
Theorems \ref{thormgtaudag} shows that $G_{\tau}$ is a DAG.
Removing minimum number of edges from a directed graph to
make it a DAG is known as the {\it minimum Feedback Arc Set} problem which is NP-hard \cite{feedbackasnp}. Thus, using 
a greedy algorithm, we try to remove as minimum number of edges 
as possible from $G_{\tau}$ in order to lose minimum information.
However, any such missing information will be recaptured in Step 9.
\\
\\
{\bf Step 2}:
In this step TripNet first apply HBUILD on $h$.
If the result is a tree, TripNet constructs a binarization of this tree. Otherwise 
TripNet goes to Step 3.
\\
\\
Note that if $\tau$ is consistent with a tree, by Theorem 
\ref{thorm2hequal}, $h_{G_{\tau}}=h_{T_{\tau}}$ and TripNet constructs
a binarization of $T_{\tau}$.
\\
\\
{\bf Step 3}: Remove all the maximum-weight edges from $G$. 
The process of removing all the maximum-weight edges from 
the graph continues until the resulting graph is disconnected.
\\
\\
In \cite{galled7} and \cite{Levelone} the authors introduced the concept of $SN$-sets
for a set of triplets $\tau$. 
A subset $S$ of $L(\tau)$ is
an {\it $SN$-set} if there is no triplet $i j|k \in \tau$ 
such that $i \notin S$ and $j,k \in S$.
In \cite{Levelone} it is shown that if $\tau$ is dense then 
the maximal $SN$-sets partition $L(\tau)$ and can be found in polynomial time.
By contracting each of the $SN$-set to a single node
and assuming a common ancestor for all of these leaves, the size of the problem is reduced. 
In these papers, for finding the
maximal $SN$-sets in polynomial time, the authors use the high density of
the input triplet sets. TripNet algorithm uses the concept of height function 
as an auxiliary tool to obtain $SN$-sets instead of the high density assumption.
\\ 
\\
{\bf Step 4}: 
For each connceted component obtaining 
in Step 3, which is not an $SN$-set we apply Step 3.
This process continues until all of the resulting components are $SN$-sets.
Let $\{S_{1}, S_{2}, \ldots, S_{k}\}$ be the set of resulting $SN$-sets.
If each $SN$-set contains only one node, 
HBUILD is applied and if the result is a tree TripNet constructs a binary tree and goes to Step 6. Otherwise TripNet goes to Step 5.
If for some $i$, $|S_{i}|>1$, contract each $S_i$ to a single node $s_i$ and set $S=\{s_1, s_2, \ldots, s_k\}$.
Update the set of triplets by defining
$\tau_{S} = \{s_i  s_j|s_k  : if ~ \exists ~ xy|z \in \tau , ~ x \in S_i, ~ y \in S_j ~$ and $~ z \in S_k \}$. 
Constructs a weighted complete graph $(G_{S}, w_{S})$
with $V(G_{S})=S$ and $w_{S}(s_{i}, s_{j})=\min \{h(x, y) : x \in S_{i} ~ and ~ y \in S_{j}\}$.
Set $(G, w)=(G_{S}, w_{S})$
and TripNet goes to Step 3.
\\
 \\
 The following theorem is consequence of definition $SN$-set for $(G_{S}, w_{S})$.
 
\,\,\,\,\,\,

\begin{theorem}
\label{thormsn} Applying Steps 3 and 4 on $(G_{S}, w_{S})$ and $\tau_{S}$, each
resulting $SN$-set has one member.
\end{theorem}

\begin{proof} 
Suppose that $S=\{s_1, s_2, s_3, \ldots, s_r\}$ is an $SN$-set  in $(G_{S}, w_{S})$. Now assume that in the procedure of Step 3 by removing the edges with weight $l$, $S_1$ separates from $S_2$. Thus there exists $k > l$ such that by removing 
the edges with weight at least $k$ in $(G_{S}, w_{S})$, the connected component $S$ seperates from other components of $G_S$.
It means that by removing the edges with weight at least $k$ in $G$, we obtain the $SN$-set $S_1 \cup \ldots \cup S_r$ in $\tau$ which is 
a contradiction.
\end{proof}

\,\,\,\,\,\,\,

In the next step the reticulation leaves are recognized using the following three criteria:

{\bf Criterion I}:
Let $m_i$ and $M_i$ be the
minimum and maximum weight of the edges in $(G,h)$ with exactly one end in 
$S_i$. 
Choose the node with minimum
$m_i$ and if there is more than one node with minimum $m_i$ then
choose among them the nodes which has minimum $M_i$. Let $R_1$ denotes the set of such nodes.

{\bf Criterion II}: Let $w_{min}= min \{w(s_{i},s_{j}) : 1 \leq i, j \leq k \}$.
In $G_{S}$ consider the induced subgraph on the edges with the weight $w_{min}$.
Choose the nodes of $R_{1}$ with the maximum degree in this induced subgraph.
Let $R_2$ denotes the set of such nodes.

{\bf Criterion III}: 
For each node $s \in R_{2}$, remove it from $G_{S}$ and find $SN$-sets
for this new graph using Steps 3 and 4. Let $n_{s}$ be the number of $SN$-sets
of this new graph with cardinality greater than one. Choose the nodes in $R_{2}$ with maximum
$n_{s}$.
Let $R_3$ denotes the set of such nodes.

We state an example to show
the idea behind these three criteria.

Let $\tau=\{ij|l, jk|i, kl|j, kl|i, no|m, lo|k, jl|o, mn|l, mn|j, no|k, mo|i, jk|n, ij|o, ik|m, il|n\}$. 
$\tau$ is not consistent with a tree but it is consistent with the network $N$ shown in Fig. \ref{netstep5}.  
Obviously, $N$ is an optimal network consistent with $\tau$. 
In order to find $SN$-sets we 
construct $G'_{\tau}$ and $(G, h)$,  
and find $SN$-sets from $(G, h)$ using Steps 3 and 4 (Figs. \ref{bothnotdaganddag} to \ref{gh1r}).
It follows that $S=\{\{i\}, \{j\}, \{k\}, \{l\}, \{m\}, \{n, o\}\}$. Now in $G_{S}$ (Fig \ref{ghu}).
we expect that the reticulation is in $R_1$. In this example
both $k$ and $l$ are in $R_1$.
Also we expect that if there is a reticulation leaf, it belongs to $R_2$ which again both $k$ and $l$ are in $R_2$.
Now just $l$ belongs to $R_3$.
Thus we consider $l$ as the reticulation leaf
 (Figs \ref{ghu4retr} to \ref{ghu3retr6543r}).
Remove triplets from $\tau_{S}$ which contain $l$ and denote the new set of triplets by $\tau'_{S}$. Obviously,
$\tau'_{S}$ is consistent with a tree. We add this reticulation leaf to a binarization of $T_{\tau'_{S}}$
such that the resulting network is consistent with $\tau_{S}$. Note that if we consider each node except than $l$
as the reticulation leaf then final network consistent with $\tau_{S}$ has at least two reticulation leaves.
\\
\begin{figure}[h]
\centering \subfigure[$N$]{
\includegraphics[scale=.4]{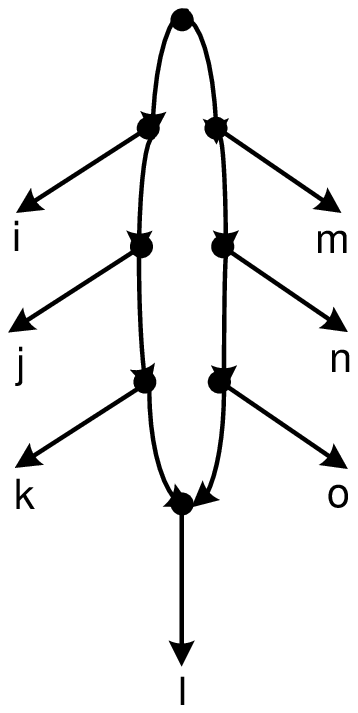}
\label{netstep5} } \subfigure[$G'_{\tau}$ is obtained from $G_{\tau}$ by removing the dotted line.]{
\includegraphics[scale=.5]{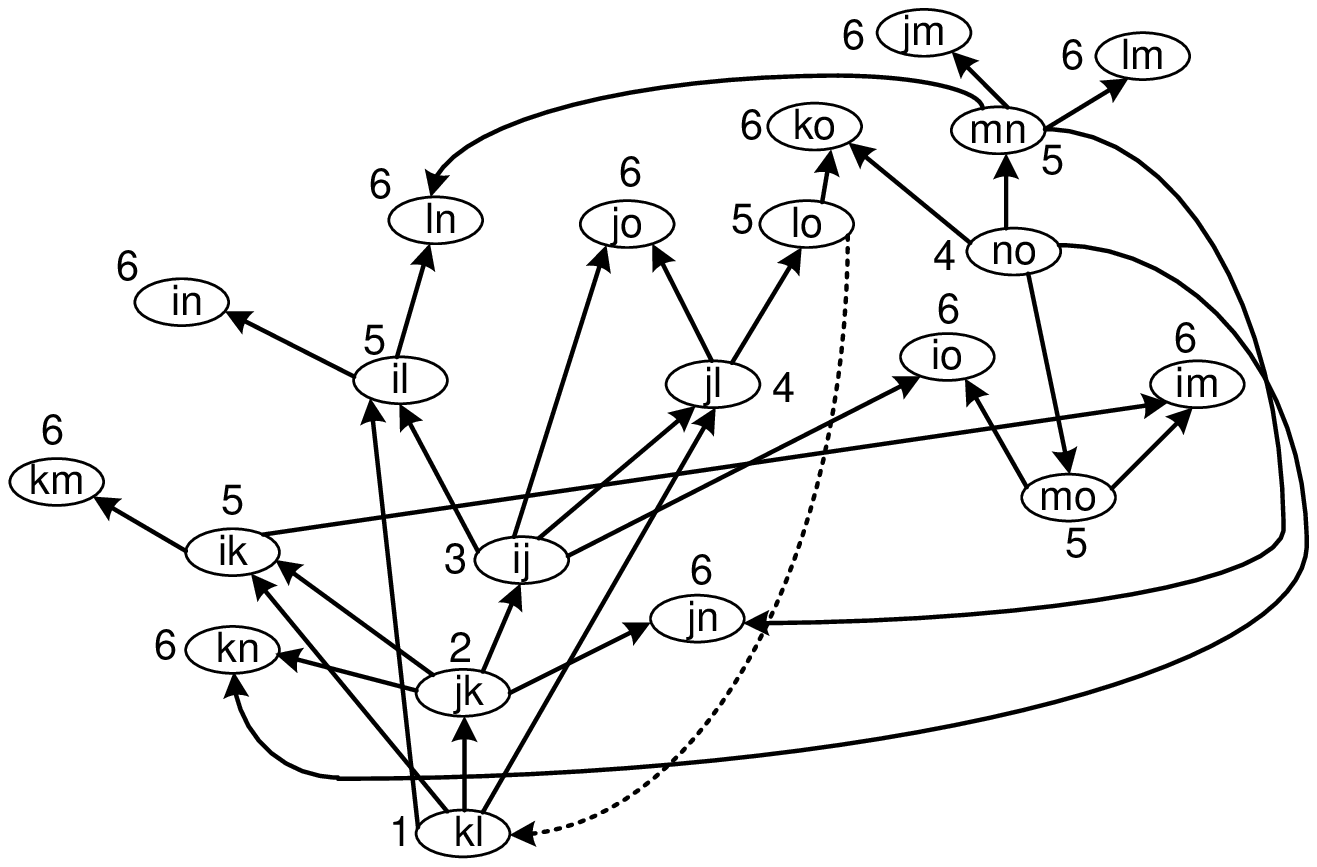}
\label{bothnotdaganddag} }  \subfigure[$(G, h)$]{
\includegraphics[scale=0.5]{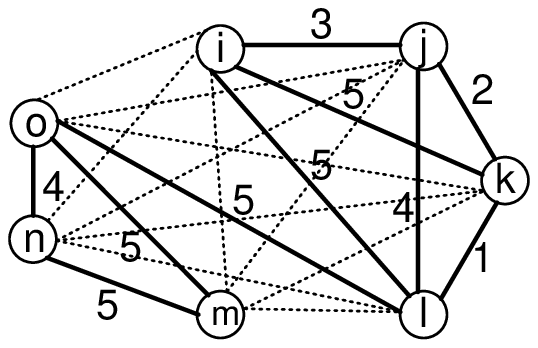}
\label{ghstep5} } \subfigure[]{
\includegraphics[scale=0.5]{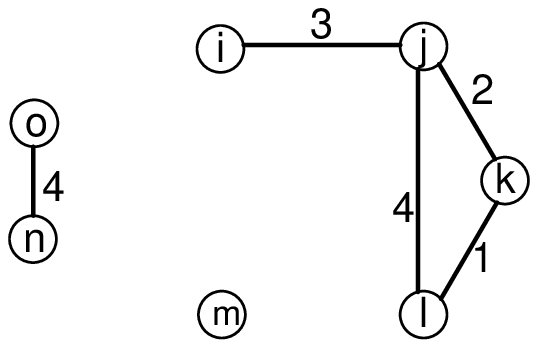}
\label{gh65r} } \subfigure[]{
\includegraphics[scale=0.5]{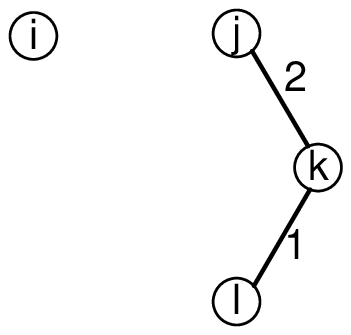}
\label{gh43r} } \subfigure[]{
\includegraphics[scale=0.5]{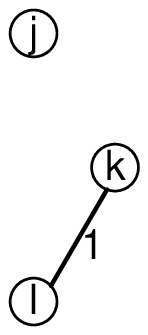}
\label{gh2r} } \subfigure[]{
\includegraphics[scale=0.5]{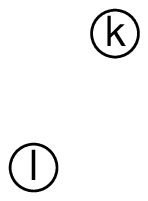}
\label{gh1r} } \subfigure[$G_{S}$]{
\includegraphics[scale=0.5]{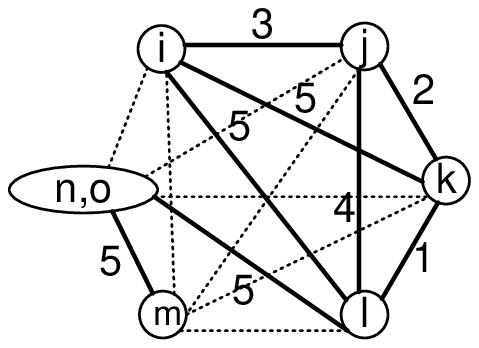}
\label{ghu} } \subfigure[]{
\includegraphics[scale=0.5]{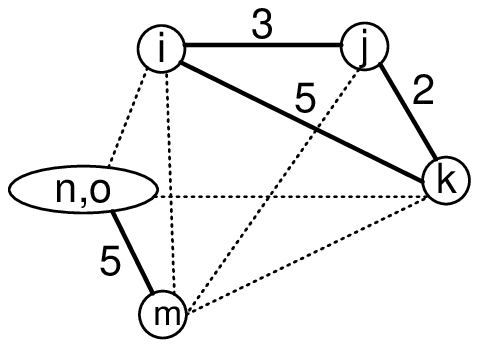}
\label{ghu4retr} } \subfigure[]{
\includegraphics[scale=0.5]{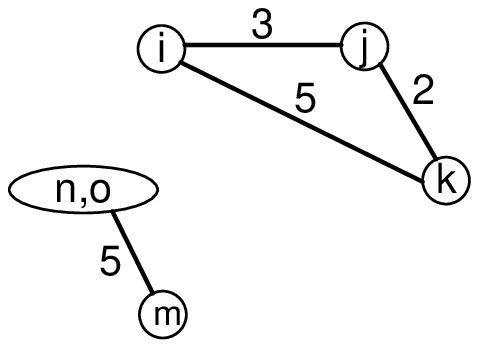}
\label{ghu4retr6r} }\subfigure[]{
\includegraphics[scale=0.5]{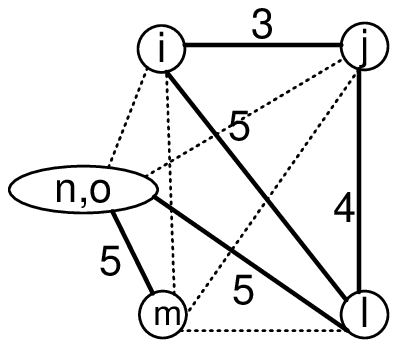}
\label{ghu3retr} }\subfigure[]{
\includegraphics[scale=0.5]{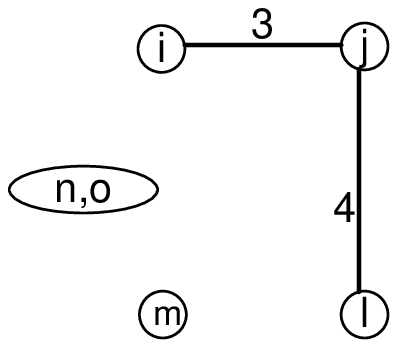}
\label{ghu3retr65r} }\subfigure[]{
\includegraphics[scale=0.5]{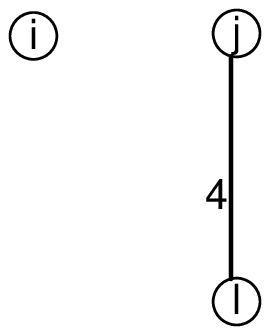}
\label{ghu3retr654r} }\subfigure[]{
\includegraphics[scale=0.5]{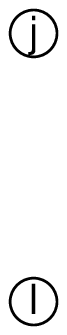}
\label{ghu3retr6543r} }
\label{step5exampleall} \caption{Edges with weight 6 are shown by dotted lines.  }
\end{figure}

{\bf Step 5} : In this step the reticulation leaf is recognized using three criteria.
Do the criterion I. If $|R_{1}|=1$ then choose the node $x \in R_{1}$ as the reticulation node.
Otherwise if $|R_{1}|>1$ do the criterion II.
If $|R_{2}|=1$ then choose the node $x \in R_{2}$ as the reticulation node.
Otherwise if $|R_{2}|>1$ do the criterion III.
If $|R_{3}|=1$ then choose the node $x \in R_{3}$ as the reticulation node.
Otherwise if $|R_{3}|>1$ then
by the speed options we choose the reticulation node as follows.

{\it Slow}: Each node in $R_{3}$ is examined as the reticulation leaf.

{\it Normal}: Two nodes in $R_{3}$ are selected randomly and each of these two nodes
is examined as the reticulation leaf.

{\it Fast}: One node in $R_{3}$ selected randomly 
as the reticulation leaf.

Let $x$ be a node which is considered as a reticulation leaf.
Remove $x$ from $G_{S}$ and all of the triplets which contain $x$ from $\tau_{S}$.
Define $G=G \setminus \{x\}$ and go to Step 3.
\\
\\
Note that for the Fast option the running time of the algorithm is polynomial.
For biological data almost always the criteria I and II find a unique reticulation leaf.
So on real data the running time of TripNet is almost always polynomial.
\\
\\
{\bf Step 6} : Let $x_1, x_2, \ldots, x_m$ be $m$ reticulation leaves which are obtained in
Step 5 with this order and $T$ be the tree that is constructed in Step 4.
Now add these $m$ nodes in the reverse order to $T$ as what follows.
Let $e_1$ and $e_2$ be two edges of $T$. Consider two new nodes $y_1$ and $y_2$ in the middle of
$e_1$ and $e_2$. Connect $y_1$ and $y_2$ to a new node $y_3$ and connect the reticulation leaf $x_m$ to $y_3$. 
Do this procedure for all  pairs of edges and
choose a pair such that the resulting network is consistent with maximum number of triplets in $\tau$.
Triplet continue this procedure until all the reticulation nodes are added.
\\
\\
{\bf Step 7} : For each $SN$-set, $S_i$, and the set $\tau_{S_{i}}$ of triplets we run the algorithm again.
\\
\\
{\bf Step 8} : Replace each $SN$-set in the network of
Step 6 with its related network constructed in Step
7 to obtain a network $N'$.
\\
\\
Let $\tau' \in \tau$ be the set of the triplets 
which are not consistent with $N'$. For each pair of leaves $a$ and $b$
assume that $\tau'_{ab}$ is the set of triplets in $\tau'$ which
are of the form $ab|c$. Consider the pair of leaves $i$ and $j$ such that
$\tau'_{ij}$ has the maximum cardinality. 
Assume that $p_{i}$ and $p_{j}$ are the parents of $i$ and $j$,
respectively. 
\\
\\
{\bf Step 9} : Create two new nodes in the middle of the edges
$p_{i}i$ and $p_{j} j$ and connect them with a new edge. This new edge creates a
reticulation node and all of the triplets in $\tau'_{ij}$ will be
consistent with the new network. 
All consistent triplets with the new network are removed from $\tau'$ and this
procedure will continue until $\tau'$ becomes empty. 
\\
\\
Fig. \ref{tripsteps}
presents an example of the algorithm with all of its Steps.

\begin{figure}[ht]
\begin{center}
\includegraphics[width=7.cm]{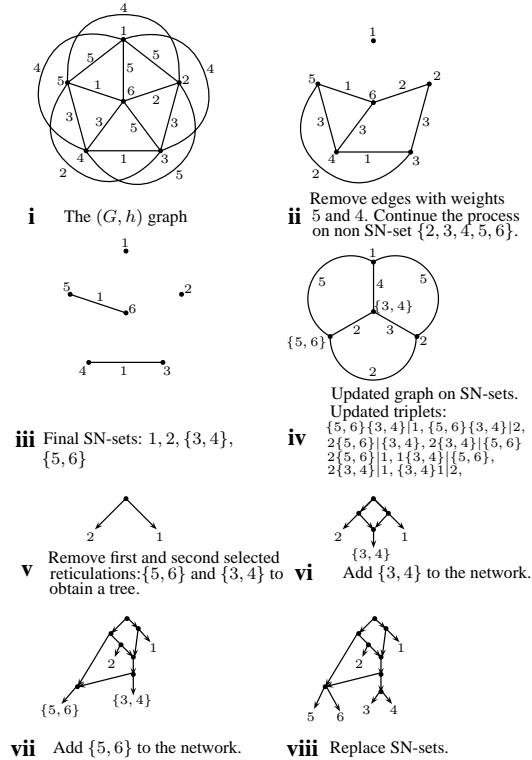}
\end{center}
\caption{Steps of our algorithm for input triplets:  $2~3|1$,
$4~ 1|2$, $5~2|1$, $2~6|1$, $3~4|1$, $1~3|5$, $1~3|6$, $5~4|1$,
 $4~6|1$, $5~6|1$, $3~4|2$, $3~5|2$, $2~6|3$, $4~5|2$, $2~4|6$, $5~6|2$,
 $3~4|5$, $3~4|6$, $5~6|3$, $5~6|4$. } \label{tripsteps}
\end{figure}

\section*{Acknowledgment}
The authors would like to thank M.Kargar for his kind and useful
comments. Changiz Eslahchi and Hadi Poormohammadi
would like to thank Shahid Beheshti University for its support. This work was partially supported by a grant from IPM.


\begin{thebibliography}{1}


\bibitem{hussonbook}
D.H. Huson, R. Rupp and C. Scornavacca, \emph{Phylogenetic
Networks Concepts, Algorithms and Applications}, Cmbridge
University Press, 2010.


\bibitem{AhoSagiv1}
A.V. Aho, Y. Sagiv, T.G. Szymanski and J.D. Ullman, "Inferring a
tree from lowest common ancestors with an application to the
optimization of relational expressions," SIAM J. Comp., vol. 10,
pp. 405-421, 1981.

\bibitem{galled7}
J. Jansson, N.B.  Nguyen and W.K. Sung, "Algorithms for combining
rooted triplets into a galled phylogenetic network," Proc.
Sixteenth annual ACM-SIAM symposium on Discrete algorithms, vol.
24, pp. 23-25, 2007.


\bibitem{Levelone}
J. Jansson and W.K. Sung, "Inferring a Level-1 Phylogenetic
Network from a Dense Set of Rooted Triplets," Theoretical Computer
Science, vol. 363, pp. 60-68 , 2006.

\bibitem{lev2optimum}
L.V. Iersel and S. Kelk, "Constructing the Simplest Possible
Phylogenetic Network from Triplets," Algorithmica, 2009.


\bibitem{levelk}
T.H. To and M. Habib, "Level-k Phylogenetic Network can be
Constructed from a Dense Triplet Set in Polynomial Time," In CPM
2009, pp. 275-288, 2009.

\bibitem{tripnetpaper}
R. Tusserkani, Ch. Eslahchi, H. Poormohammadi and A. Azadi,
"TripNet: A Method for Constructing Phylogenetic Networks from
Triplets," arXiv:1104.4720v1, 2011.


\bibitem{iersel11}
L.V. Iersel, J. Keijsper, S. Kelk, L. Stougie, F. Hagen and T.
Boekhout, "Constructing level-2 phylogenetic networks from
triplets," In RECOM 2008, vol. 4955, pp. 450-462, 2008.

\bibitem{fourqunique}
H-J. Bandelt and M.A Steel, "Symmetric matrices representable by weighted
trees over a cancellative abelian monoid," SIAM, J. Disc. Math, vol. 8, No. 4, pp 517-525, 1995.

\bibitem{feedbackasnp}
R. Karp, "Reducibility among combinatorial problems," Proc.
Sympos., IBM Thomas J. Watson Res. Center, Yorktown Heights, N.Y.,
pp. 85-103, 1972.


\end{thebibliography}
\end{document}